\pdfoutput=1

\documentclass[11pt]{article}

\usepackage{acl}

\usepackage{times}
\usepackage{latexsym}

\usepackage[T1]{fontenc}

\usepackage[utf8]{inputenc}

\usepackage{microtype}

\usepackage{times}
\usepackage{graphicx} 
\usepackage{natbib} 
\usepackage{caption}

\usepackage{latexsym,booktabs}
\usepackage{multirow}
\usepackage{amsthm,amssymb}
\newtheorem{theorem}{Theorem}

\usepackage{amsmath}
\usepackage{subcaption}

\theoremstyle{definition}
\newtheorem{definition}{Definition}[section]
\usepackage{comment}
\usepackage{tikz}
\usepackage{tikzscale}
\usepackage[ruled,vlined]{algorithm2e}
\usepackage{pgfplots}
\usetikzlibrary{positioning}
\usetikzlibrary{backgrounds}
\pgfplotsset{width=10cm,compat=1.9}
\usepgfplotslibrary{groupplots}
\usetikzlibrary{pgfplots.statistics}

\DeclareMathOperator*{\argmin}{argmin}

\DeclareMathOperator*{\mmax}{max}

\newcommand*\samethanks[1][\value{footnote}]{\footnotemark[#1]}
\title{The Impact of Differential Privacy on Group Disparity Mitigation}

\author{
Victor Petr\'en Bach Hansen \textsuperscript{\rm 1 \rm 2}\thanks{\;Authors contributed equally},
Atula Tejaswi Neerkaje\textsuperscript{\rm 3 \rm 4}\samethanks ,
Ramit Sawhney \textsuperscript{\rm 3}, \\
\textbf{Lucie Flek \textsuperscript{\rm 3},
Anders S\o gaard}\textsuperscript{\rm 1} \\
\textsuperscript{\rm 1}Department of Computer Science, University of Copenhagen, Denmark\\ 
\textsuperscript{\rm 2}Topdanmark A/S, Denmark\\
\textsuperscript{\rm 3}Conversational AI and Social Analytics (CAISA) Lab, University of Marburg, Germany \\
\textsuperscript{\rm 4}Manipal Institute of Technology, India\\
}

\begin{document}
\maketitle
\begin{abstract}
The performance cost of differential privacy has, for some applications, been shown to be higher for minority groups; fairness, conversely, has been shown to disproportionally compromise the privacy of members of such groups. Most work in this area has been restricted to computer vision and risk assessment. In this paper, we evaluate the impact of differential privacy on fairness across four tasks, focusing on how attempts to mitigate privacy violations and between-group performance differences interact: Does privacy inhibit attempts to ensure fairness? To this end, we train $(\varepsilon,\delta)$-differentially private models with empirical risk minimization and group distributionally robust training objectives. Consistent with previous findings, we find that differential privacy increases between-group performance differences in the baseline setting; but more interestingly, differential privacy {\em reduces} between-group performance differences in the robust setting. We explain this by reinterpreting differential privacy as regularization. 
\end{abstract}

\section{Introduction}

Classification tasks in computer vision and natural language processing face the challenge of
balancing performance with the need to prevent discrimination against protected demographic subgroups, satisfying fairness principles. In some tasks, we train our classifiers on private data and therefore also need our models to satisfy privacy guarantees. 


Privacy-preserving algorithms, however, may tend to disproportionally affect members of minority classes \citep{farrand2020neither}. E.g., \citet{bagdasaryan2019differential}, show the performance cost of differential privacy \cite{10.1007/11681878_14} in face recognition is higher for minority groups, suggesting that privacy and fairness are fundamentally at odds \citep{chang2021privacy,agarwal21tradeoffs}. In this paper, we evaluate two hypotheses at scale: (a) that the performance cost of differential privacy is unevenly distributed across demographic groups \citep{pmlr-v81-ekstrand18a,cummings2019compatibility,bagdasaryan2019differential,farrand2020neither}, and (b) that such effects can be mitigated by more robust learning objectives \citep{sagawa*2020distributionally,pezeshki2020gradient}. 

\paragraph{Contributions} We build upon previous work suggesting that differential privacy and fairness are at odds: Differential privacy may often hurt minority groups the most, and reducing the fairness gap by focusing on minority groups during training typically puts their privacy at risk. We evaluate this hypothesis at scale by measuring the impact of differential privacy in terms of fairness across (1) a baseline empirical risk minimization and (2) under a group distributionally robust optimization. We conduct our experiments across four tasks of different modalities, assuming the group membership information is available at training time, but not at test time: face recognition (CelebA), topic classification, volatility forecasting based on earning calls, and sentiment analysis of product reviews. Our results confirm that differential privacy compromises fairness in the baseline setting;
however, we demonstrate that 
differential privacy not only mitigates the decrease but also {\em improves } fairness compared to non-private experiments for 4/5 datasets in the distributionally robust setting. We explain this by reinterpreting differential privacy as an approximation of Gaussian noise injection, which is equivalent to strategies previously shown to determine the efficacy of group-robust learning.




\section{Fairness and Privacy}

Fair machine learning aims to ensure that induced models do not discriminate
against individuals with specific values in their protected
attributes (e.g., race, gender). We represent each data point
as $z = (x, g, y) \in \mathcal{X} \times \mathcal{G} \times \mathcal{Y}$, with $g\in \mathcal{G}$ encoding its protected attribute(s).\footnote{In practice our protected attributes in \S3 will be {\em age} and {\em gender}. Both are protected under the Equality Act 2010.} Let $\mathcal{D}^g_y$ denote the distribution of data with protected attribute $g$ and label $y$. 

Several definitions of group fairness exist in the literature \citep{pmlr-v97-williamson19a}, but here we focus on a generalization of approximately constant conditional (equalized) risk \citep{donini18empirical}:\footnote{In the fairness literature, approximate fairness is referred to as $\delta$-fairness, but below we will use lower case $\delta$ to refer to $(\varepsilon,\delta)$-differential privacy, and we refer to $\Delta$-fairness to avoid confusion.} 

\begin{definition}[$\Delta$-Fairness] Let $\ell^{g_i}(\theta)=\mathbb{E}[\ell(\theta(x), y)|g=g_i]$ be the risk of the samples in
the group defined by $g_i$, and $\Delta\in [0, 1]$. We say that a model $\theta$ is $\Delta$-fair if for any two values of $g$, say $g_i$ and $g_j$,  $|\ell^{g_i}(\theta)-\ell^{g_j}(\theta)|<\Delta$.
\end{definition}

Note that if $\ell$ coincides with the performance metric of a task, and $\delta=0$, this is identical to performance or classification parity \citep{yuan2021assessing}.\footnote{Performance or classification parity has been argued to suffer from statistical limitations in \cite{corbettdavies2018measure}, which remind us that when risk distributions
differ, standard error metrics are poor proxies of individual equity. This is known as the problem of infra-marginality. Note, however, that this argument does not apply to binary classification problems.} Such a notion of fairness can
be derived from John Rawls’ theory on
distributive justice and stability, treating model performance as a resource to be allocated. Rawls' {\em difference principle}, maximizing the welfare
of the worst-off group, 
is argued to lead to stability and mobility in society at large \citep{rawls_theory_1971}. $\Delta$ directly measures what is sometimes called Rawlsian {\em min-max fairness} \citep{10.1287/opre.1100.0865}.  In our experiments, we measure $\Delta$-fairness as the absolute difference between performance of the worst-off and best-off subgroups.

Recall the standard definition of  $(\varepsilon,\delta)$-privacy: 

\begin{definition} $\theta$ is $(\varepsilon,\delta)$-private iff 
    $\mbox{Pr}[\theta(\mathcal{X})]\leq \exp(\varepsilon)\times  \mbox{Pr}[\theta(\mathcal{X}')]+\delta$ 
for any two distributions, $\mathcal{X}$ and $\mathcal{X'}$, different at most in one row.\end{definition}

Differential privacy thereby ensures that an algorithm will generate similar outputs
on similar data sets. Note the multiplicative bound $\exp(\varepsilon)$ and the additive bound $\delta$ serve different roles: The $\delta$ term represents the possibility that a few data points are not governed by the multiplicative bound, which controls the level of privacy (rather than its scope). Note that it also follows directly that if $\varepsilon=0$ and $\delta=0$, absolute privacy is required, leading $\theta$ to be independent of the data. 

Several authors have shown that differential privacy comes at different costs for minority subgroups \citep{pmlr-v81-ekstrand18a,cummings2019compatibility,bagdasaryan2019differential,farrand2020neither}. The more private the model is required to be, the larger group disparities it may exhibit\footnote{Note this is a different trade-off than the fairness-privacy trade-off which results from the need for collecting sensitive data to learn fair models; the latter is discussed at length in \citet{veale2017fairer}.}. This happens because differential privacy distributes noise where it is needed to reduce the influence of individual examples. Since outlier examples are likely to have disproportional influence on output distributions \citep{10.2307/2347160,Chernick1983TheUO}, they are also disproportionally affected by noise injection in differential privacy. 

\citet{agarwal21tradeoffs} show that, in fact, a $(\varepsilon,0)$-private and fully fair model -- using equalized odds as the definition of fairness -- will be unable to learn anything. To see this, remember that a fully private model is independent of the data and unable to learn from correlations between input and output. If $\theta$ is, in addition, required to be fair, it is thereby required to be fair for all distributions, which prevents $\theta$ from encoding any prior beliefs about the output distribution. Note this finding generalizes straight-forwardly to equalized risk, and even to approximate fairness (since even for finite distributions, we can define a $\Delta>0$, such that preserving absolute privacy would lead to a constant $\theta$). 

\begin{theorem}
For sufficiently small values of $\Delta$, a fully $(\varepsilon,0)$-private model $\theta$ that is also $\Delta$-fair,  will have trivial performance. 
\end{theorem}

\begin{proof}
This follows directly from the above. 
\end{proof}

While we do not strictly require an absolute privacy in our experiments (setting $\delta=10^{-5}$), intuitively, privacy potentially compromises fairness by adding more noise to data points of minority group members than to those of majority groups. Fairness, on the other hand, leads to over-sampling or over-attending to data points of minority group members, more likely compromising their privacy. 

\citet{pannekoek2021investigating} show, however, that it {\em is} possible to learn {\em somewhat private} and {\em somewhat fair} classifiers. They combine differential privacy with reject option classification. Their results nevertheless suggest that privacy and fairness objectives are fundamentally at odds, as fairness decreases with the introduction of differential privacy. 

\section{Experiments}
\begin{figure*}[t]
    \centering
    \includegraphics[width=\textwidth]{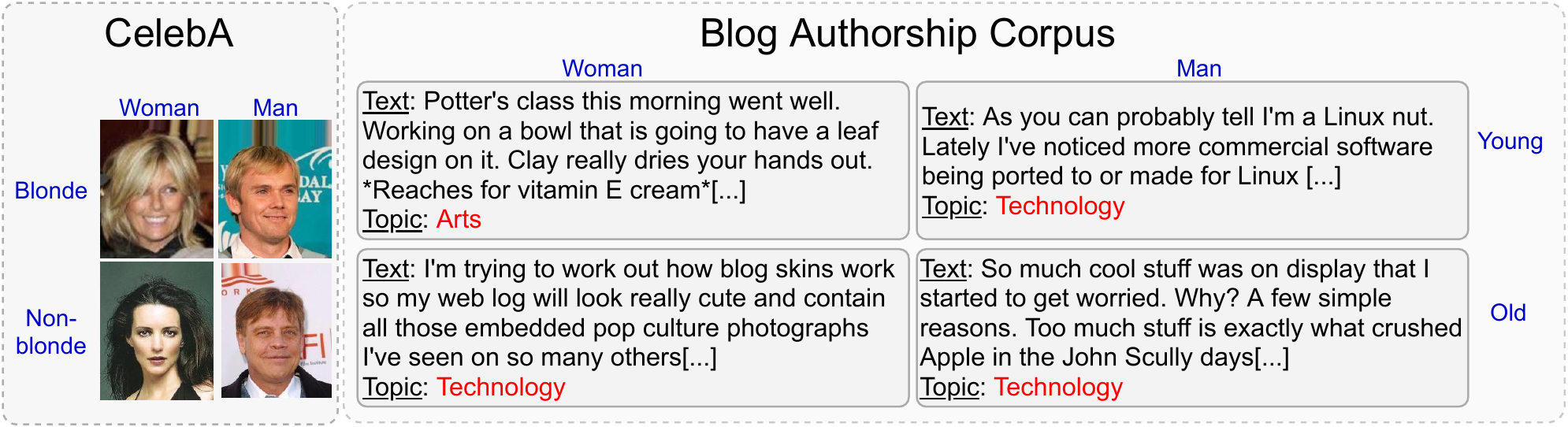}
    \caption{Examples of the different subgroups that appear in a subset of the datasets we train on. CelebA (left) contains images of celebrities, using hair-color as our target variable and gender as our protected attribute.  Blog Authorship Corpus (right) contains text-based blogposts on two topics \{Technology, Arts\} our targets, using $\mathcal{G}:\{\mbox{Man}, \mbox{Woman}\}\times \{\mbox{Young}, \mbox{Old}\}$ as our protected subgroups.}
    \label{fig:data}
\end{figure*}

This section describes the algorithms and datasets involved in our experiments, and presents the results of these. 
\subsection{Algorithms}


\paragraph{Empirical Risk Minimization} For a model parameterized by $\theta$, in our baseline Empirical Risk Minimization (ERM) setting, we minimize the expected loss $\mathbb{E}[\ell(\theta(x), y)]$ with data $(x,g,y)\in\mathcal{X}\times \mathcal{G}\times \mathcal{Y}$ drawn from a dataset $\mathcal{D}$:
\begin{equation}
    \hat{\theta}_{ERM}= \argmin_{\theta} \mathbb{E}_{\hat{\mathcal{D}}}[\ell(\theta(x), y)]    
\end{equation}
Here $\hat{\mathcal{D}}$ denotes the empirical training distribution. Note that we disregard any group information in our data. In an overparameterized setting, ERM is prone to overfitting spurious correlations, which are more likely to hurt performance on minority groups \citep{sagawa2020investigation}. 

\paragraph{Distributionally Robust Optimization} Several authors have suggested 
to mitigate the effects of such overfitting by 
explicitly optimizing for out-of-distribution mixtures of sub-populations \citep{hu18robust,orenSHL19,sagawa*2020distributionally}. In this work we focus on Group-aware Distributionally Robust Optimization (Group DRO)  \citep{sagawa*2020distributionally}. 

Under the assumption that the training distribution $\mathcal{D}$ is a mixture of a discrete number of groups, $\mathcal{D}_g$ for $g\in \mathcal{G}$, we define the worst-case loss as the maximum of the group-specific expected losses:
\begin{equation}
    \ell(\theta)_{worst}=\mmax_{g\in \mathcal{G}} \mathbb{E}_{\hat{\mathcal{D}_g}}[\ell(\theta(x), y)]
\end{equation}
In Group DRO -- in contrast with ERM -- we exploit our knowledge of the group membership of data points ($x,g,y$). The overall objective is for minimizing the empirical worst-case loss is therefore: 
\begin{equation}
    \hat{\theta}_{DRO}=\argmin_{\theta} \Big[ \ell(\hat{\theta})_{worst}:=\mmax_{g \in G} \mathbb{E}_{\hat{\mathcal{D}}_g}[\ell(\theta(x), y)] \Big]
\end{equation}
Note, again, that the knowledge of group membership $g$ is only available at training time, not at test time. Unlike \citet{sagawa*2020distributionally}, we do not employ heavy $\ell_2$ regularization during our experiments, but rather use it with the same parameters as proposed in \citet{koh2021wilds}.



\paragraph{Differentially Private Stochastic Gradient Descent (DP-SGD)} We implement differential privacy \citep{10.1007/11681878_14} using DP-SGD, as presented in \citet{abadi2016deep}. DP-SGD limits the influence of training samples by (i) clipping the per-batch gradient where its norm exceeds a pre-determined clipping bound $C$, and by (ii) adding Gaussian noise $\mathcal{N}$ characterized by a noise scale $\sigma$ to the aggregated per-sample gradients. We control this influence with a privacy budget $\varepsilon$, where lower values for $\varepsilon$ indicates a more strict level of privacy. DP-SGD has remained popular, among other things because it generalizes to iterative
training procedures \citep{48293}, and supports tighter bounds using the Rényi method \citep{8049725}. 

Differential privacy generally comes at a performance cost, leading to privacy-preserving models performing worse compared to their non-private counterparts \citep{alvim2011differential}. However, we follow \citet{Kerrigan2020DifferentiallyPL} and {\em finetune} the private models, which are first pretrained (without differential privacy) on a large public dataset. This protocol generally seems to provide a better trade-off between accuracy and privacy \citep{Kerrigan2020DifferentiallyPL}, leading to better-performing, yet private models. The only exception to this setup is the volatility forecasting task, where our models were trained from scratch, as those rely on PRAAT audio features.

\begin{table*}[t]
    \centering
    \scriptsize
    \begin{tabular}{clcccccccc}
        \toprule 
        &&\multicolumn{8}{c}{Performance at $\varepsilon$-Privacy}\\
        && \multicolumn{2}{c}{No DP} & \multicolumn{2}{c}{$\varepsilon_1$} & \multicolumn{2}{c}{$\varepsilon_2$} & \multicolumn{2}{c}{$\varepsilon_3$} \\
        && Score & $\varepsilon$ & Score & $\varepsilon$ & Score & $\varepsilon$ & Score & $\varepsilon$ \\
               \midrule 
        \multirow{2}{*}{\rotatebox[origin=c]{90}{\scriptsize \sc Celeb}}&
        {\sc ERM}&$0.954 \pm 0.000$&-&$0.943 \pm 0.001$&$9.50$&$0.940 \pm 0.002$&$5.17$&$0.932 \pm 0.001$&$0.99$\\
        \cmidrule{2-10}&
        {\sc DRO}&$0.953 \pm 0.001$&-&$0.899 \pm 0.006$&$9.50$&$0.891 \pm 0.014$&$5.17$&$0.873 \pm 0.007$& $0.99$
        \\
             \midrule 
        \multirow{3}{*}{\rotatebox[origin=c]{90}{\scriptsize  \sc Blog}}&
        {\sc ERM}&$0.699 \pm 0.002$&-&$0.661 \pm 0.003$&$9.25$&$0.661 \pm 0.003$&$5.03$&$0.648 \pm 0.005$&$1.02$\\
        \cmidrule{2-10}&
        {\sc DRO}&$0.692 \pm 0.001$&-&$0.651 \pm 0.001$&$9.25$&$0.650 \pm 0.005$&$5.03$&$0.630 \pm 0.003$&$1.02$
        \\
        \midrule 
        \multirow{2}{*}{\rotatebox[origin=c]{90}{\sc \scriptsize Vol.}}&
        {\sc ERM}&$0.756 \pm 0.036$&-&$0.778 \pm 0.073$&$9.32$&$0.794 \pm 0.046$&$6.42$&$0.778 \pm 0.039$&$0.96$\\
        \cmidrule{2-10}&
        {\sc DRO}&$0.814 \pm 0.061$&-&$0.798 \pm 0.042$&$9.32$&$0.815 \pm 0.056$&$6.42$&$0.833 \pm 0.093$&$0.96$
        \\
        \midrule 
        \multirow{2}{*}{\rotatebox[origin=c]{90}{\scriptsize \sc T-UK}}&
        {\sc ERM}&$0.933 \pm 0.008$&-&$0.919 \pm 0.002$&$9.39$&$0.916 \pm 0.001$&$4.94$&$0.889 \pm 0.009$&$1.02$\\
        \cmidrule{2-10}&
        {\sc DRO}&$0.931 \pm 0.004$&-&$0.893 \pm 0.006$&$9.39$&$0.873 \pm 0.015$&$4.94$&$0.820 \pm 0.015$&$1.02$
        \\
        \midrule 
        \multirow{2}{*}{\rotatebox[origin=c]{90}{\scriptsize \sc T-US}}&
        {\sc ERM}&$0.894 \pm 0.007$&-&$0.817 \pm 0.014$&$10.71$&$0.812 \pm 0.009$&$5.10$&$0.666 \pm 0.019$&$1.01$\\
        \cmidrule{2-10}&
        {\sc DRO}&$0.899 \pm 0.009$&-&$0.569 \pm 0.132$&$10.71$&$0.437 \pm 0.112$&$5.10$&$0.342 \pm 0.012$&$1.01$
        \\

        \bottomrule 
    \end{tabular}
    
    \begin{tabular}{clcccccccc}
        \toprule 
        &&\multicolumn{8}{c}{Group-disparity at $\varepsilon$-Privacy}\\
        && \multicolumn{2}{c}{No DP} & \multicolumn{2}{c}{$\varepsilon_1$} & \multicolumn{2}{c}{$\varepsilon_2$} & \multicolumn{2}{c}{$\varepsilon_3$} \\
        && GD & $\varepsilon$ & GD & $\varepsilon$ & GD & $\varepsilon$ & GD & $\varepsilon$ \\
        \midrule 
        \multirow{2}{*}{\rotatebox[origin=c]{90}{\scriptsize \sc Celeb}}&
        {\sc ERM}&$0.556 \pm 0.021$&-&$0.746 \pm 0.032$&$9.50$&$0.734 \pm 0.025$&$5.17$&$0.770 \pm 0.013$&$0.99$\\
        \cmidrule{2-10}&
        {\sc DRO}&$0.514 \pm 0.042$&-&$0.039 \pm 0.018$&$9.50$&$0.080 \pm 0.031$&$5.17$&$\mathbf{0.056} \pm 0.027$&$0.99$
        \\
        \midrule 
        \multirow{3}{*}{\rotatebox[origin=c]{90}{\scriptsize  \sc Blog}}&
        {\sc ERM}&$0.108 \pm 0.013$&-&$0.149 \pm 0.006$&$9.25$&$0.140 \pm 0.004$&$5.17$&$0.136 \pm 0.011$&$0.99$\\
        \cmidrule{2-10}&
        {\sc DRO}&$0.078 \pm 0.009$&-&$0.056 \pm 0.020$&$9.25$&$0.070 \pm 0.013$&$5.17$&$\mathbf{0.077} \pm 0.027$&$0.99$
        \\
        \midrule 
        \multirow{2}{*}{\rotatebox[origin=c]{90}{\sc \scriptsize Vol.}}&
        {\sc ERM}&$0.302 \pm 0.042$&-&$0.328 \pm 0.067$&$9.32$&$0.557 \pm 0.050$&$6.42$&$0.573 \pm 0.050$&$0.96$\\
        \cmidrule{2-10}&
        {\sc DRO}&$0.221 \pm 0.062$&-&$0.320 \pm 0.085$&$9.32$&$0.371 \pm 0.058$&$6.42$&$0.421 \pm 0.083$&$0.96$
        \\
        \midrule 
        \multirow{2}{*}{\rotatebox[origin=c]{90}{\scriptsize \sc T-UK.}}&
        {\sc ERM}&$0.018 \pm 0.005$&-&$0.022 \pm 0.006$&$9.39$&$0.020 \pm 0.014$&$4.94$&$0.037 \pm 0.006$&$1.02$\\
        \cmidrule{2-10}&
        {\sc DRO}&$0.030 \pm 0.008$&-&$0.030 \pm 0.004$&$9.39$&$0.039 \pm 0.023$&$4.94$&$\mathbf{0.025} \pm 0.010$&$1.02$
        \\
        \midrule 
        \multirow{2}{*}{\rotatebox[origin=c]{90}{\scriptsize \sc T-US}}&
        {\sc ERM}&$0.055 \pm 0.006$&-&$0.048 \pm 0.019$&$10.71$&$0.054 \pm 0.015$&$5.10$&$0.109 \pm 0.017$&$1.01$\\
        \cmidrule{2-10}&
        {\sc DRO}&$0.036 \pm 0.007$&-&$0.118 \pm 0.040$&$10.71$&$0.078 \pm 0.030$&$5.10$&$\mathbf{0.021} \pm 0.030$&$1.01$
        \\

        \bottomrule 
    \end{tabular}
    \caption{Performance (top) and $\Delta$-Fairness (bottom) of ERM and Group DRO across different degrees of differential privacy ($\varepsilon$). $\varepsilon_1$, $\varepsilon_2$ and $\varepsilon_3$ corresponds to $\varepsilon$-values of roughly 10, 5 and 1 respectively (see table for exact values). We report F1 scores for sentiment and topic classification, accuracy for face recognition and MSE for volatility forecasting. Group disparity (GD) is measured by the absolute difference between the best and worst performing sub-group ($\Delta$-Fairness; see Definition~2.1). The performance and corresponding uncertainties are based on several individual runs of each configuration, see \S\ref{sec:train_details} in the Appendix for further details. Differential privacy consistently hurts fairness for ERM. For Group DRO, we {\bf bold-face} numbers where strict differential privacy ($\varepsilon_3$) {\em increases} fairness; this happens in 4/5 datasets. We see large increases for face recognition and small increases for topic classification and sentiment analysis.}
    \label{tab:main_res}
\end{table*}

\subsection{Tasks and architectures}
To study the impact of differential privacy on fairness, in ERM and Group DRO, we evaluate increasing levels of differential privacy across five datasets that span four tasks and three  different modalities: speech, text and vision.

\paragraph{Facial Attribute Detection} We study facial attribute recognition with the CelebFaces Attributes Dataset (CelebA) \citep{Liu2015face}.
It contains faces of celebrities annotated with attributes, such as hair color, gender and other facial features. Following \citet{sagawa*2020distributionally}, we use the hair color as our target variable, with gender being the demographic attribute (see Figure \ref{fig:data} (left)). The dataset contains $\sim 163K$ datapoints, where the smallest group (blond males) only counts $1387$.
We finetune a publicly pretrained ResNet50, a standard model for image classification tasks, on the CelebA dataset and evaluate model performances as accuracies over 3 individual seeds.

\paragraph{Topic Classification} For topic classification, we use the Blog Authorship Corpus \citep{Schler2006EffectsOA}.
The Blog Authorship Corpus contains weblogs written on 19 different topics, collected from the Internet before August 2004. The dataset contains \textit{self-reported} demographic information about the gender and age of the authors. The dataset is limited to gender information in the binary form\footnote{Note that all binary gender assumptions in this work are unwillingly inherited from the datasets.}. We binarize age, distinguishing  between young ($=<35$) and older ($>35$) authors,\footnote{Older authors tend to be underrepresented in web data. 
} resulting in four different group combinations (see Figure \ref{fig:data} (right)).
We chose two topics of roughly equal size (Technology and Arts), reducing the topic classification task to a binary classification task.
For our experiments, we finetune a pretrained English DistilBERT model \citep{sanh2019distilbert}. To reduce the overall added computational cost of DP-SGD, we freeze our model, except for the outer-most Transformer encoder layer as well as the classification layer. We report model performances as F1 scores over 3 individual seeds.

\paragraph{Volatility Forecasting} For the stock volatility forecasting task, we use the Earnings Conference Calls dataset by \citet{qin-yang-2019-say}. This consists of 559 public earnings calls audio recordings for 277 companies in the S\&P 500 index, spanning over a year of earnings calls. The \textit{self-reported} genders of the CEOs was scraped by \citet{sawhney-etal-2021-empirical} from Reuters,\footnote{\url{https://www.thomsonreuters.com/en/profiles.html}} Crunchbase,\footnote{\url{https://www.crunchbase.com/discover/people}} and the WikiData API.\footnote{\url{https://query.wikidata.org/}} The extracted genders were found to be in the binary form \cite{sawhney-etal-2021-empirical}, with 12.3\% of speakers being female and 87.7\% of speakers being male, a highly skewed distribution. Since our primary focus with this task is to explore the impact of differential privacy on speech, we use only audio features without the call transcripts. For each audio recording $A$ of a given earning call $E$, the goal is to predict the company's stock volatility as a regression task. Following 
\citet{qin-yang-2019-say}, we calculate the average log volatility $\tau$ days (temporal window) following the day of the earnings call. For each audio clip belonging to a given call, we extract 26-dimensional features with PRAAT \citep{boersma2001speak}. Each audio embedding of the call is fed sequentially to a BiLSTM,  
followed by an attention layer and two fully-connected layers. The model is trained by optimizing the Mean Square Error (MSE) between the predicted and true stock volatility. For all results, we report MSE on the test set for a 70:10:20 temporal split of the data.
The results are averaged over 5 seeds. 

\paragraph{Sentiment Analysis} For our sentiment analysis task, we use the Trustpilot Corpus \citep{hovy2015user}\footnote{\url{https://bitbucket.org/lowlands/release/src/master/WWW2015/data/}}. It consists of text-based user reviews from the Trustpilot website, rating companies and services on a 1 to 5 star scale. The reviews spans 5 different countries; Germany, Denmark, France, United Kingdom and USA, however, we only consider the English reviews, i.e. UK and US. The Trustpilot contains demographic information about the gender, age and geographic location of the users, but as with the topic classification task, we only concern ourselves with the gender and age of the users. 
As with the topic classification task, we finetune DistilBERT on the UK and US English parts of the Trustpilot Corpus, freezing all parameters but the final encoder layer, as well as the classification layer. Classification performance is measured as F1 scores and the results are averaged over 3 seeds.

Our implementation is a PyTorch extension of the WILDS repository\footnote{\url{https://github.com/p-lambda/wilds/}} \citep{koh2021wilds} using the DP-SGD implementation provided by the Opacus Differential Privacy framework\footnote{\url{https://opacus.ai/}}. For further details about data and training, see \S\ref{sec:train_details} in the Appendix. We release the code for our experiments at: \url{https://github.com/vpetren/fair_dp}.

\begin{figure*}[h]
\centering
\begin{subfigure}{.5\textwidth}
    \centering
    \includegraphics[width=\textwidth]{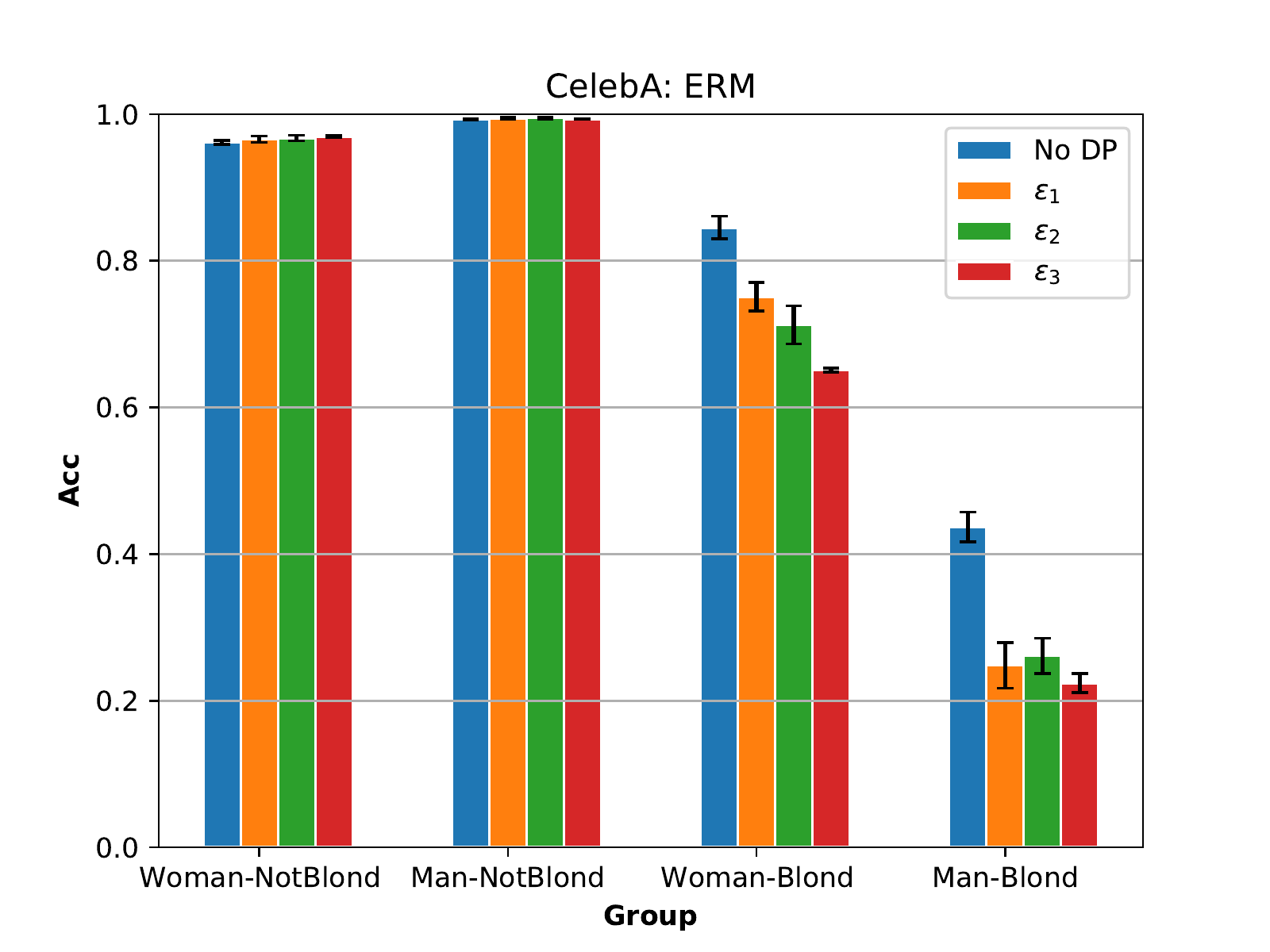}
\end{subfigure}%
\begin{subfigure}{.5\textwidth}
    \centering
    \includegraphics[width=\textwidth]{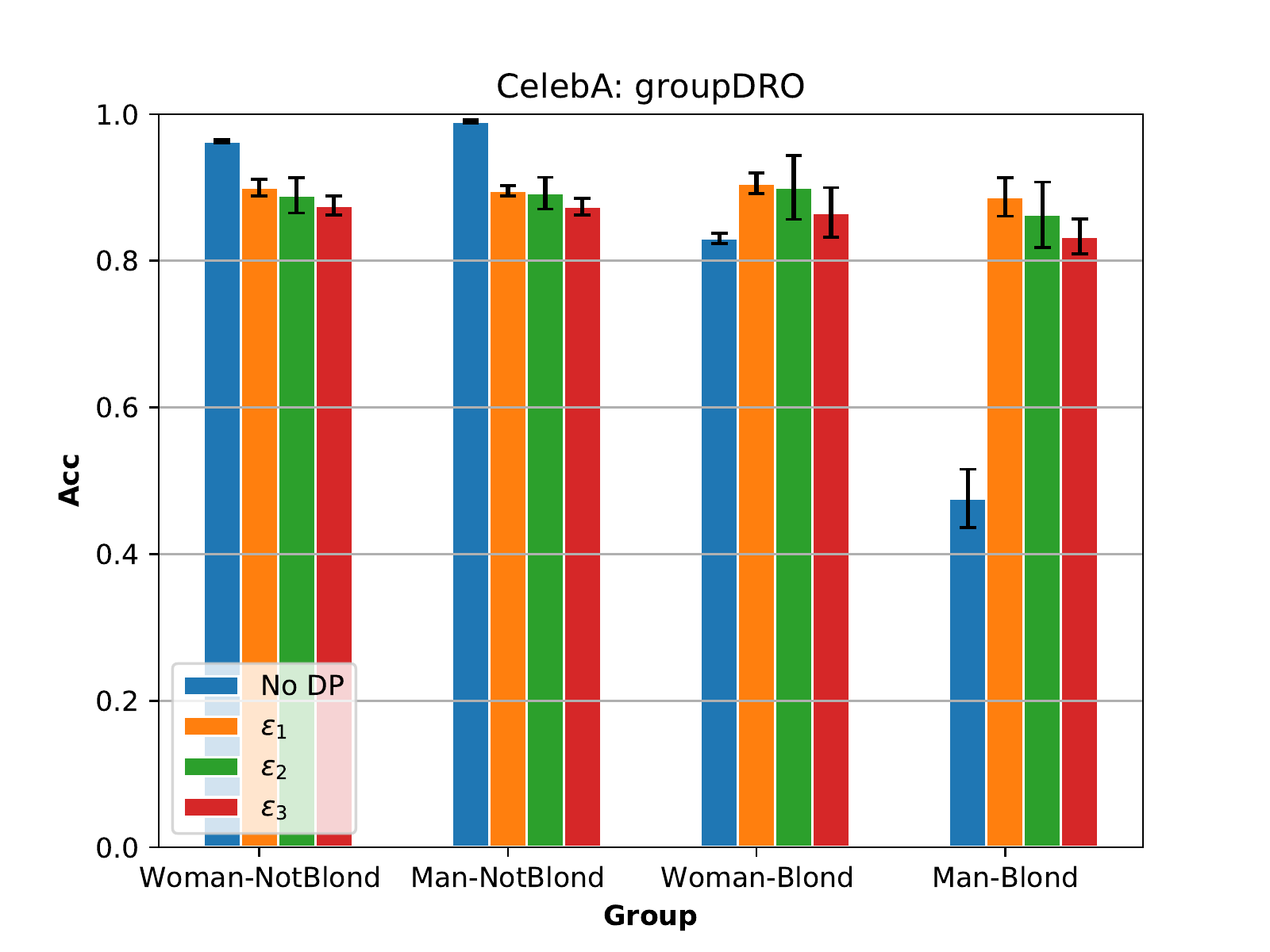}
\end{subfigure}
\caption{{\bf Face Attribute Detection:} Performance of individual groups of increasing levels of $\varepsilon$. Comparing baseline ERM to Group DRO, we find that Group DRO performance on the minority group (blond males) perform much better under privacy constraints; we return to this in \S3.4.}
\label{fig:barplot-face}
\end{figure*}

\begin{figure*}[h]
\centering
\begin{subfigure}{.5\textwidth}
    \centering
    \includegraphics[width=\textwidth]{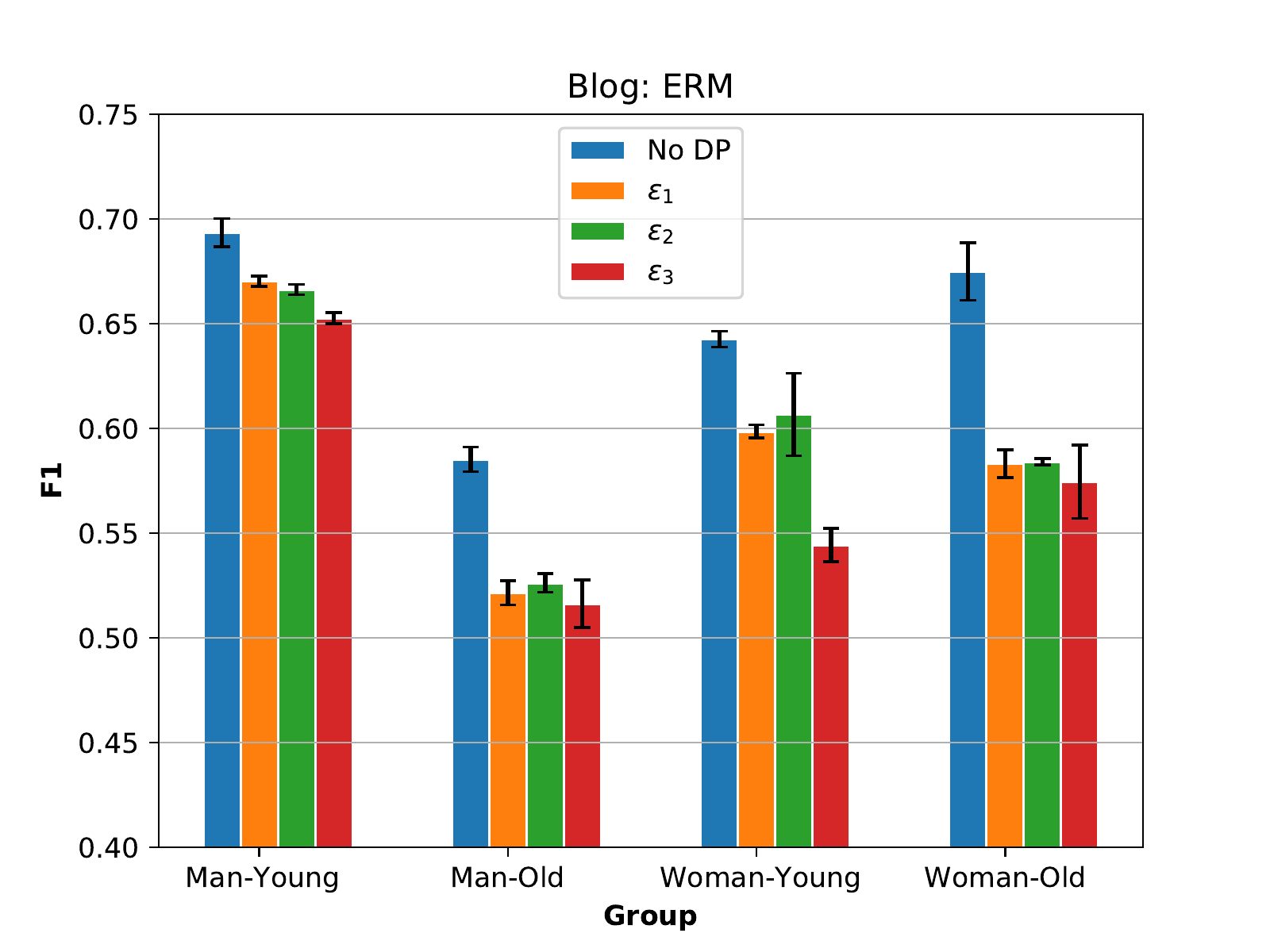}
\end{subfigure}%
\begin{subfigure}{.5\textwidth}
    \centering
    \includegraphics[width=\textwidth]{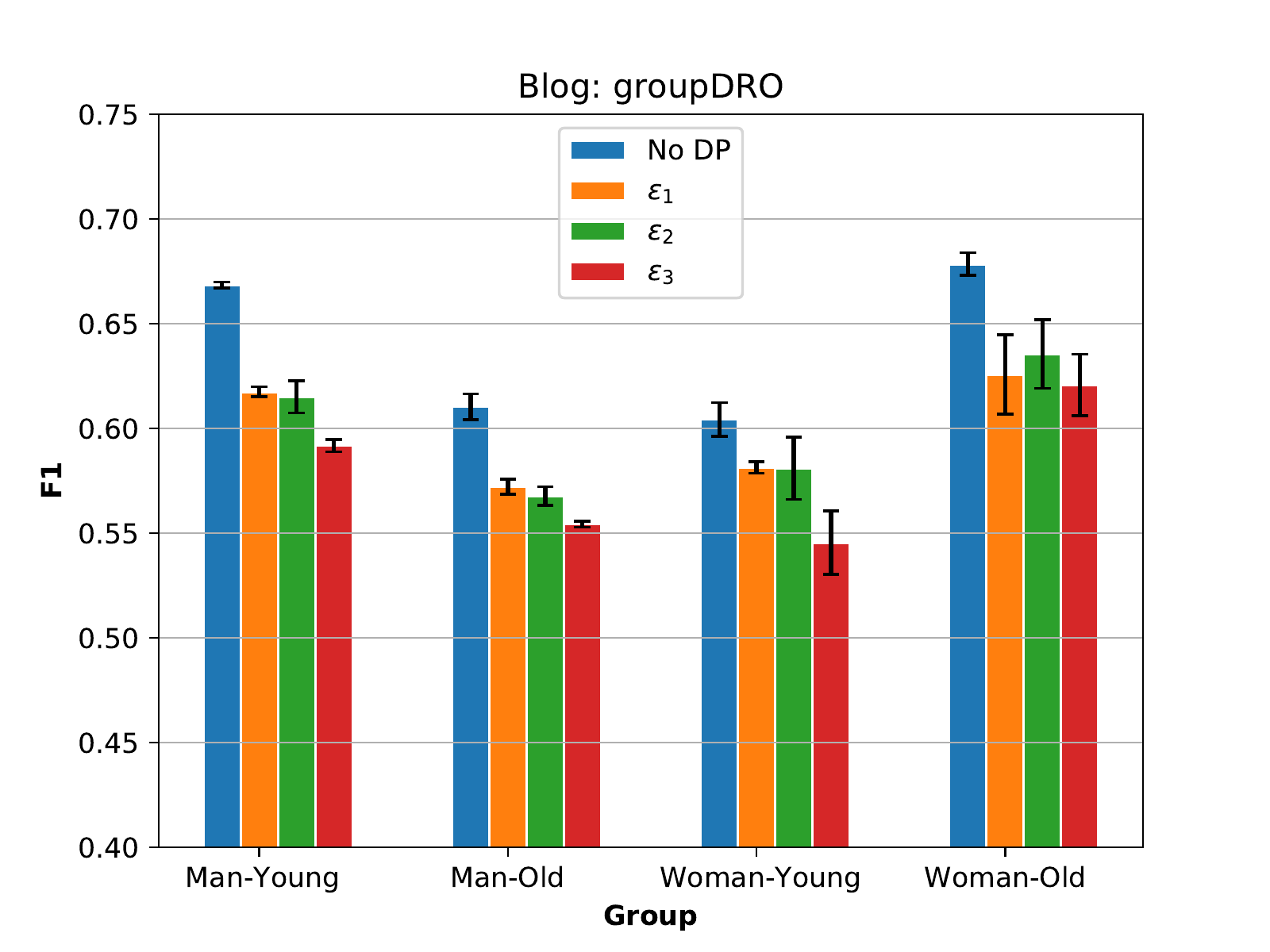}
\end{subfigure}
\caption{{\bf Topic Classification:} Performance of individual groups of increasing levels of $\varepsilon$. 
 Group DRO, compared to baseline ERM, results in a more balanced performance across all groups, even on a low privacy budget. }
\label{fig:barplot-topic}
\end{figure*}

\subsection{Results}

Our results are presented in Table \ref{tab:main_res}. The top half of the table presents standard (average) performance numbers across multiple runs of ERM and Group DRO at different privacy levels. Recall that performance for sentiment analysis as well as topic classification is measured in F1, volatility forecasting is measured in MSE and face recognition is measured in accuracy.  The accuracy of our ERM face attribute detection classifier is 0.954 in the non-private setting, for example. 

Our first observation is that, as hypothesized earlier, differential privacy hurts model performance. For our smallest text-based dataset (\textsc{T-US}), performance becomes very poor at the strictest privacy level. This is however associated with a high amount of variance between seeds, see Figure \ref{fig:group_trust_us} in the Appendix. The above face attribute detection classifier, which had an accuracy of 0.954 in the non-private setting, has a performance of 0.932 at this level. 

\paragraph{Differential privacy hurts fairness in ERM} The effect on differential privacy on fairness (bottom half of Table~\ref{tab:main_res}) is also quite consistent. The gap between the majority group and the minority group (or, more precisely, the best-performing and the worst-performing demographic subgroup) widens with increased privacy. In face recognition, for example, the accuracy gap between the two groups is 0.556 without differential privacy, but 0.770 at the strictest privacy level. 

\begin{figure*}[t]
\centering
\begin{subfigure}{\textwidth}
    \centering
    \hspace{-1.5cm}
    \includegraphics[width=0.55\linewidth,height=5cm]{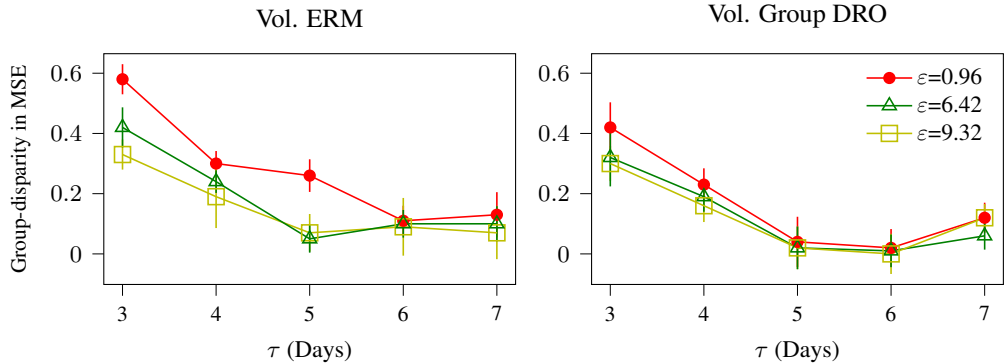}
\end{subfigure}
    \caption{{\bf Volatility Forecasting: }A comparison of group-disparity between subgroups for increasing temporal volatility windows ($\tau$) and privacy budgets ($\varepsilon$), over 5 independent runs.}
    \label{fig:TemporalDisparity}
\end{figure*}

\paragraph{Differential privacy increases fairness in Group DRO} For Group DRO, we see the opposite effect. For 4/5 datasets, we see that differential privacy leads to an increase in fairness. For face recognition, for example, the gap goes from 0.514 in the non-private setting to 0.056 in the strictest, basically disappearing. This is also illustrated in the bar plots in Figure~\ref{fig:barplot-face}. See Figure~\ref{fig:barplot-topic} for similar bar plots of the topic classification results; we include similar plots for other tasks in the Appendix. We do also observe that this increase in privacy can be expensive in terms of overall performance (e.g. Trustpilot-US). Note that the increase in fairness at higher privacy levels is seemingly at odds with previous results suggesting that privacy and fairness conflict, e.g., \citet{agarwal21tradeoffs}. We return to this question in \S3.4.

Note also that the only exception to the latter trend is for volatility forecasting, where differential privacy hurts fairness both in ERM and Group DRO (though Group DRO mitigates the disparity). This speech-based prediction is the only regression task, and the only task for which we do not rely on pretrained models trained on public data.

For this task, we further analyze group disparity for varying temporal windows ($\tau$) used to calculate target volatility values, along with increasingly strict privacy budgets ($\varepsilon$) in Figure \ref{fig:TemporalDisparity}. The disparity between subgroups widens with stricter privacy guarantees \citep{bagdasaryan2019differential}. This gap is significant for lower values of $\tau$, strengthening the hypothesis that short-term volatility forecasting is much harder than long-term \citep{qin-yang-2019-say}, especially for minority classes due to the disproportionate impact of noise. Comparing ERM and Group DRO, we find Group DRO mitigates this disparity gap. We observe disparity reduces with increasing temporal window, since stock prices over a larger time frame are comparatively more stable \citep{qin-yang-2019-say}. As a consequence, the influence of Group DRO for higher $\tau$ ($6,7$) is reduced, despite facilitating faster convergence. Most importantly, we observe the power of Group DRO in mitigating the disparity caused by strict privacy safeguards ($\varepsilon=0.96$) for crucial short term prediction ($\tau=3)$ tasks. 

\subsection{Discussion}

It is well-known that differential privacy comes with a performance cost 
\citep{shokri2015privacy}.\footnote{A multitude of algorithmic improvements have been proposed to mitigate the overall accuracy drop caused by the increased privacy protection -– including private sampling from hyperbolic word representation spaces \citep{feyisetan2019leveraging}, 
Gaussian $f$-differential privacy (Bu et al. 2020), and gradient denoising (Nasr et al., 2020). It is yet to be examined, if the empirical application of such utility preservation techniques affects the disparate impact issue.} However, recent work has additionally shown that differential privacy is at odds with most, if not all, definitions of fairness, including equalized risk \citep{pmlr-v81-ekstrand18a,cummings2019compatibility,bagdasaryan2019differential,farrand2020neither}. Our work makes two important contributions: (a) We evaluate and confirm this hypothesis at a larger scale than previous studies for standard empirical risk minimization; and (b) we point out that the opposite holds true in the context of Group Distributionally Robust Optimization: Here, adding differential privacy improves fairness (equalized risk). 

While (b) at first seems to contradict the very hypothesis that (a) confirms -- namely that privacy is at odds with fairness -- we believe the explanation is quite simple, namely that we are observing two opposite trends (at the same time): On one hand, differential privacy adds disproportionate noise to minority group examples; but on the other hand, it adds Gaussian noise which acts as a regularizer to improve robust optimization. 

In their evaluation of Group Distributionally Robust Optimization, \citet{sagawa*2020distributionally} observe that robustness is only achieved in the context of heavy regularization; specifically, they show fairness improvements when they add $\ell_2$ regularization or early stopping. The $\ell_2$ regularization and early stopping did not increase fairness under ERM, but seemed to 'activate' Group DRO. This makes intuitive sense: Since regularized models cannot perfectly fit the training data, heavily regularized Group DRO sacrifices average performance for worst-case performance and obtain better generalization. In the absence of regularization, however, Group DRO is less effective. 

In our experiments (\S3), we add minimal regularization to Group DRO, following the implementation in \citet{koh2021wilds}, but differential privacy, we argue, provides that additional regularization. To see this, remember that DP-SGD works by Gaussian noise injection. Gaussian noise injection is known to be near-equivalent to $\ell_2$-regularization and early stopping \citep{bishop95training}. DP-SGD simply makes the trade-off more urgent.  


\section{Related Work}

\paragraph{Fair machine learning} Early work on mitigating group-level disparities included oversampling \citep{shen2016relay,guo2004learning} and undersampling \citep{drumnond2003class,barandela2003restricted}, as well as instance weighting \citep{Shimodaira2000}. Other proposals modify existing training algorithms or cost functions to obtain fairness \citep{khan2017cost,chung2015cost}. In the context of large-scale deep neural networks, Group DRO is a particularly interesting approach to mitigating group-level disparities \citep{creager2021environment}.  See \citet{pmlr-v97-williamson19a} and \citet{corbettdavies2018measure} for interesting discussions of how fairness has been measured. More recent alternatives to Group DRO include Invariant Risk Minimization \citep{arjovsky2020invariant}, Spectral Decoupling \citep{pezeshki2020gradient} and Adaptive Risk Minimization \citep{zhang2021adaptive}. We ran experiments with both Invariant Risk Minimization and Spectral Decoupling, but they performed much worse than Group DRO. 

\paragraph{Fairness and privacy} Recent studies suggest that privacy-preserving methods such as differential privacy tend to disproportionately affect minority class samples \citep{pmlr-v81-ekstrand18a,cummings2019compatibility,bagdasaryan2019differential,farrand2020neither}.
\citet{pannekoek2021investigating} show that it {\em is} possible to learn {\em somewhat private} and {\em somewhat fair} classifiers, in their case by combining differential privacy and reject option classification. \citet{jagielski2019differentially} introduced the so-called DP-oracle-learner, derived from an \textit{oracle-efficient} algorithm \citep{agarwal2018reductions}, which satisfies equalized odds, an alternative notion of fairness \citep{pmlr-v97-williamson19a}. \citet{lyu2020democratise} introduced Differentially Private GANs (DPGANs), while \citet{tran2020differentially} utilize Lagrangian duality to integrate fairness constraints to protected attributes. Group DRO has, to the best of our knowledge, not been studied under differential privacy before.

\section{Conclusions} 
In \S2, we summarized previous work suggesting that differential privacy and fairness are at odds. In \S3, we then confirmed this hypothesis at scale, across five datasets, spanning four tasks and three modalities, showing that for Empirical Risk Minimization, stricter levels of privacy consistently {\em hurt} fairness. This holds true even after pretraining on large-scale public datasets \citep{Kerrigan2020DifferentiallyPL}. In the context of Group-aware Distributionally Robust Optimization (Group DRO) \citep{sagawa*2020distributionally}, however, which is designed to mitigate group-level performance disparities (optimizing for equalized risk), we saw the opposite effect: Strict levels of differential privacy were associated with an {\em increase} in fairness. In \S3.4, we discuss how this aligns well with the observation that Group DRO works best in the context of heavy $\ell_2$ regularization, keeping in mind that Gaussian noise injection is near-equivalent to $\ell_2$ regularization \citep{bishop95training}.

\bibliography{anthology,custom}
\bibliographystyle{acl_natbib}

\clearpage
\newpage

\section{Appendix}
\label{sec:appendix}
\subsection{Additional Figures}
This section contains group-specific bar-plots for the performance on individual groups in the Trustpilot Corpus. For barplots on CelebA and Blog Authorship, see Figure \ref{fig:barplot-face} and \ref{fig:barplot-topic}.
\begin{figure*}[h]
\centering
\begin{subfigure}{.5\textwidth}
    \centering
    \includegraphics[width=\textwidth]{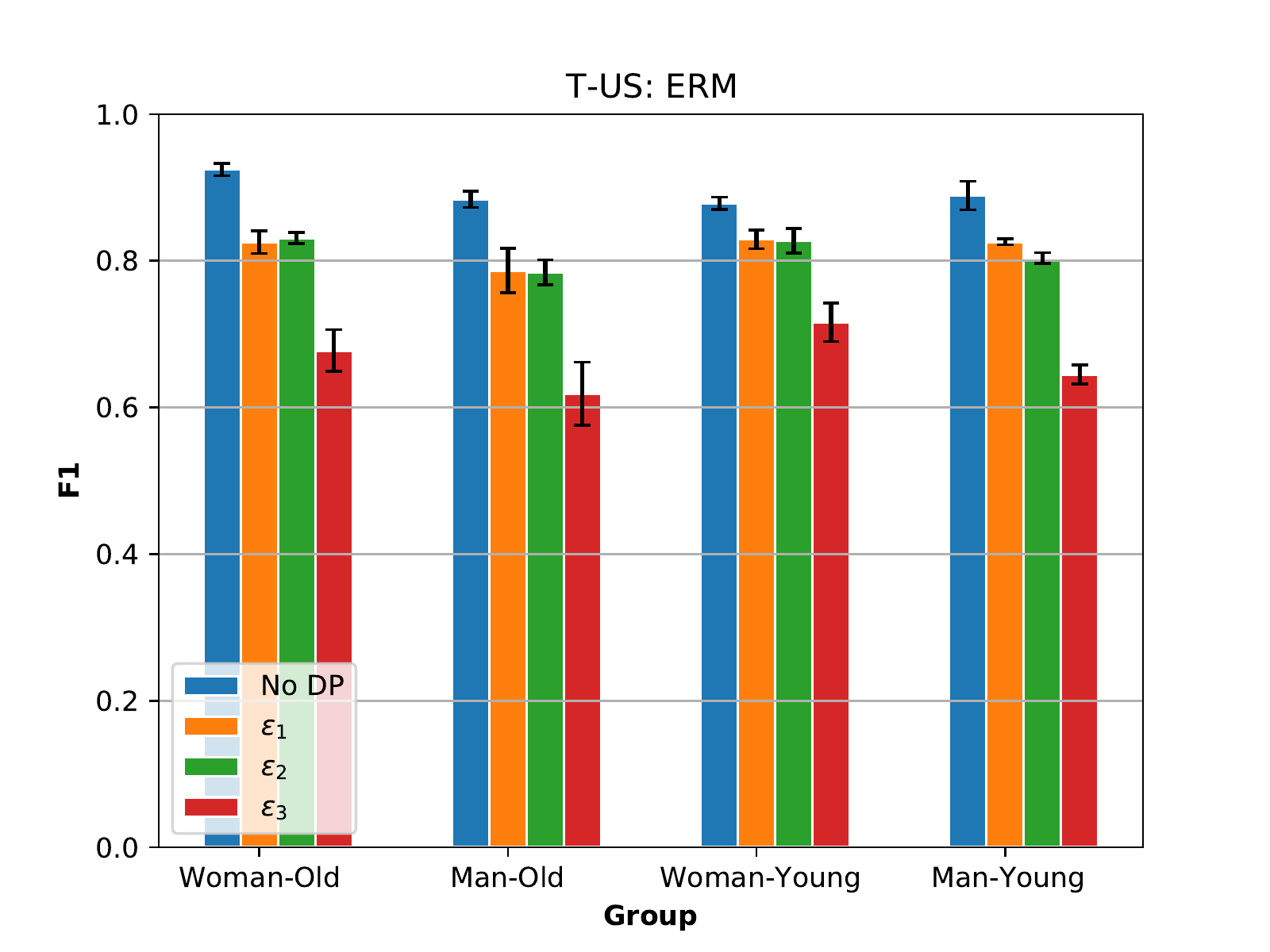}
\end{subfigure}%
\begin{subfigure}{.5\textwidth}
    \centering
    \includegraphics[width=\textwidth]{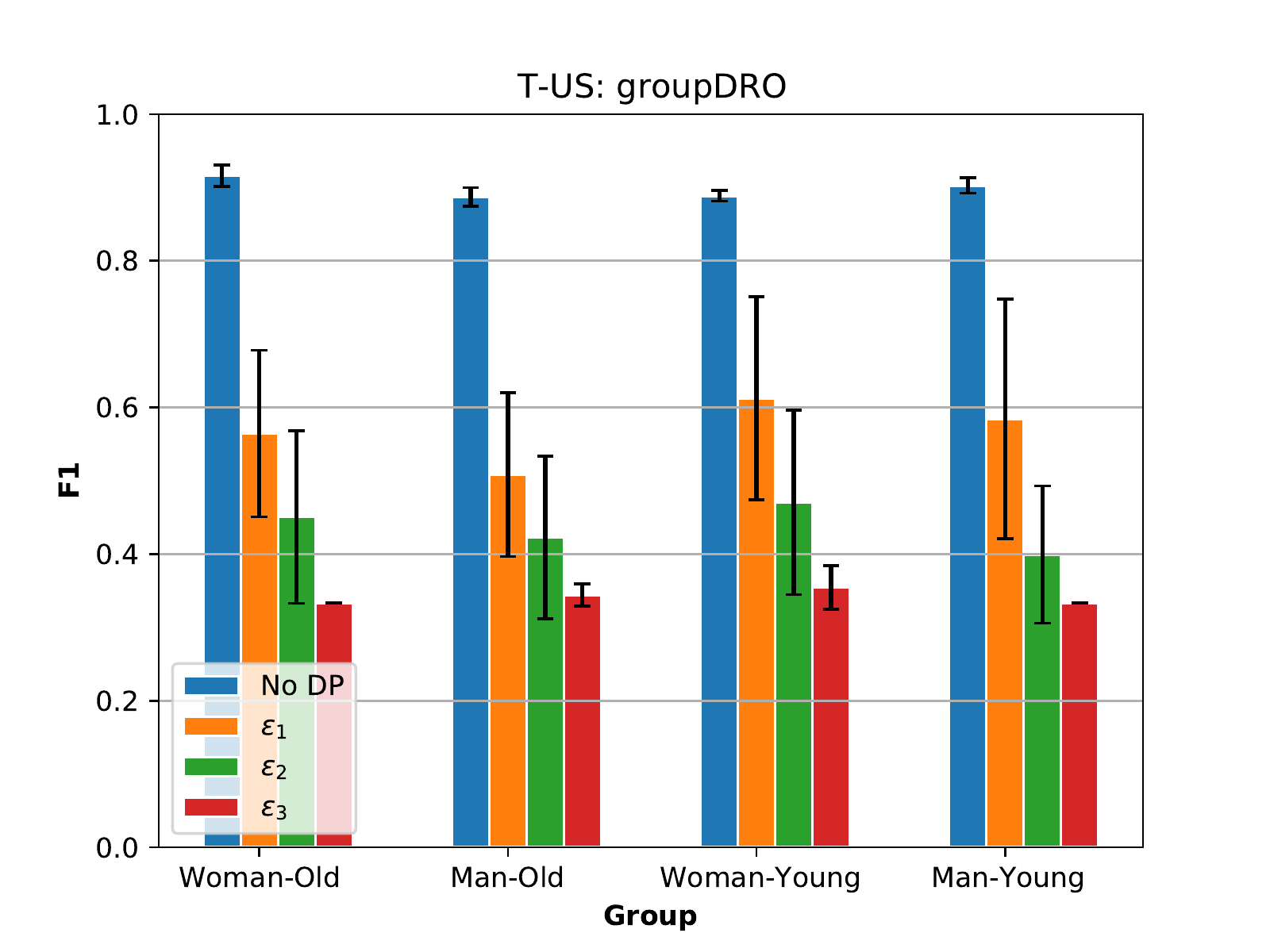}
\end{subfigure}
\caption{Performance of individual groups of increasing levels of $\varepsilon$ for the Trustpilot-US corpus. Error bars show standard deviation over 3 individual seeds.}
\label{fig:group_trust_us}
\end{figure*}

\begin{figure*}[h]
\centering
\begin{subfigure}{.5\textwidth}
    \centering
    \includegraphics[width=\textwidth]{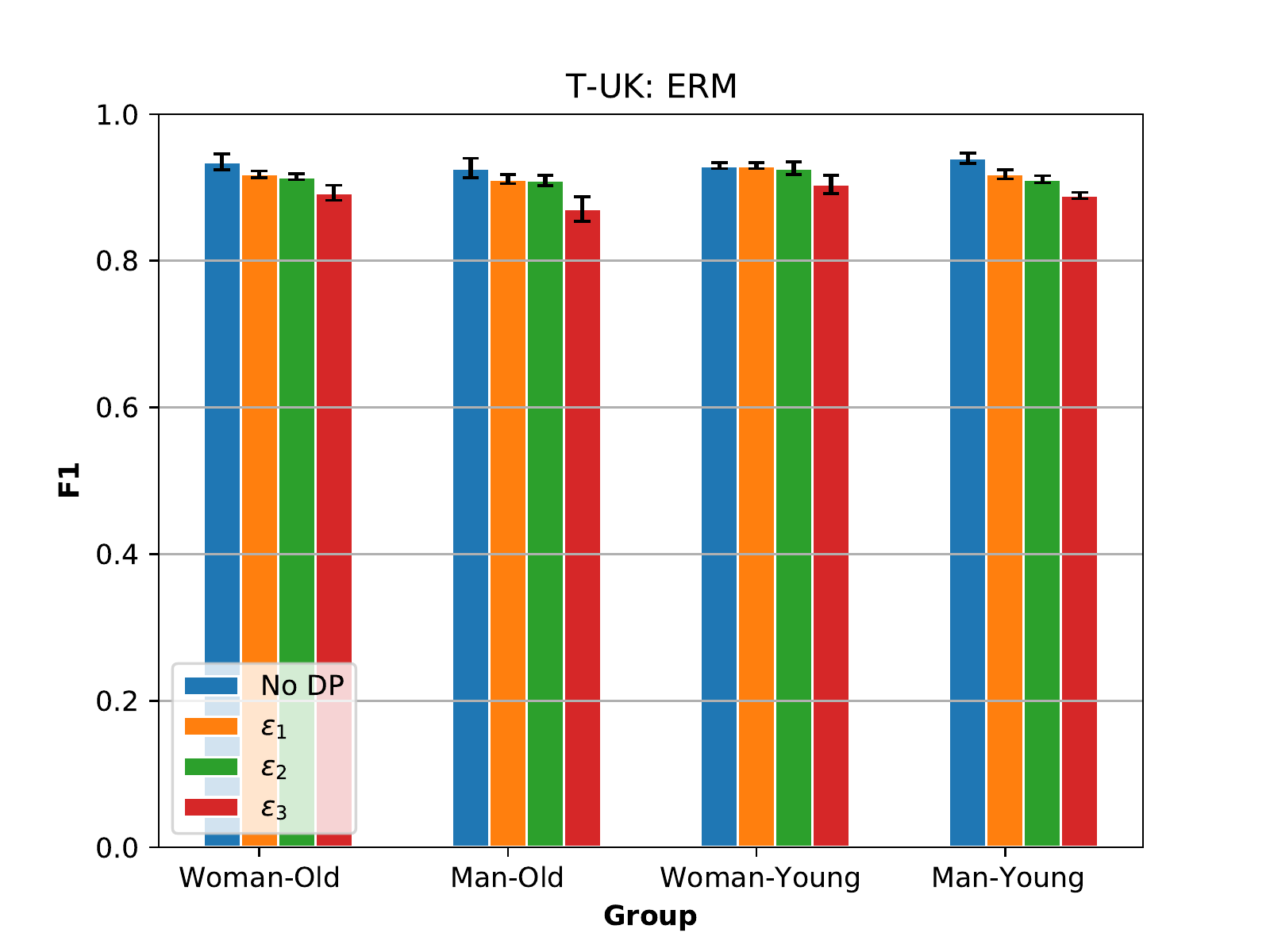}
\end{subfigure}%
\begin{subfigure}{.5\textwidth}
    \centering
    \includegraphics[width=\textwidth]{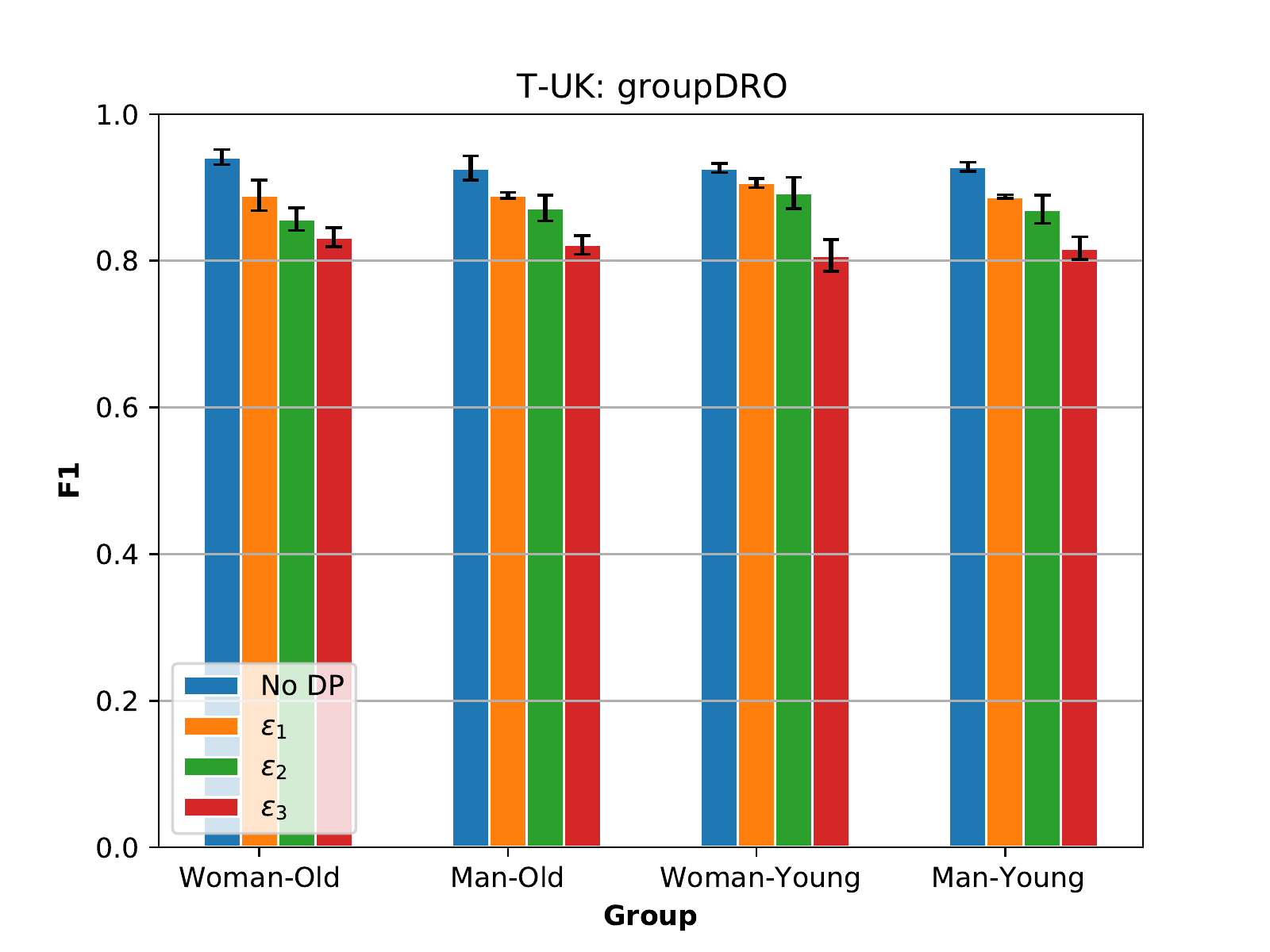}
\end{subfigure}
\caption{Performance of individual groups of increasing levels of $\varepsilon$ for the Trustpilot-UK corpus. Error bars show standard deviation over 3 individual seeds.}
\label{fig:group_trust_uk}
\end{figure*}

\subsection{Experimental Details}
\label{sec:train_details}
This section contains additional details surrounding the experiments described in \S3.
\paragraph{CelebA} We use the same processed version of the CelebA dataset as \citet{sagawa*2020distributionally} and \citet{koh2021wilds}, that is, we use the same train/val/test splits as \citet{Liu2015face} with the \textit{Blond Hair} attribute as the target with the \textit{Male} attribute being the spuriously correlated variable. See group distribution in the training data in Table \ref{tab:dist_celeb}.


\begin{table}[h]
    \scriptsize
    \centering
    \begin{tabular}{|c|c|c|c}
         Non-Blond, Man & Blond, Man & Non-Blond, Woman & Blond, Woman\\
        \hline
        $66874 $ & $1387$ & $71629$ & $22880$ \\
    \end{tabular}
    \caption{Group distribution in the training set of CelebA}
    \label{tab:dist_celeb}
\end{table}

\paragraph{Blog Authorship Corpus} In addition to the preprocessing described in \S3, we split the data into a 60/20/20 train/val/test split (you can find the exact seed that generates the splits in our code). See group distribution in the training data in Table \ref{tab:dist_blog}.
\begin{table}[h]
    \scriptsize
    \centering
    \begin{tabular}{|c|c|c|c|c|}
        Group & Young, Man & Old, Man & Young, Woman & Old, Woman\\
        \hline
        Count & $27222$ & $2295$ & $12750$ & $2435$ \\
    \end{tabular}
    \caption{Group distribution in the training set of Blog Authorship corpus}
    \label{tab:dist_blog}
\end{table}
The Blog Authorship Corpus can be downloaded at: \url{https://www.kaggle.com/rtatman/blog-authorship-corpus}

\paragraph{Earnings Conference Calls} Out of the 559 calls, we only include 535 datapoints that contain self-reported demographic attributes about gender. See Table \ref{tab:dist_ecalls} for group distributions for the training data. The target stock volatility variable is calculated following \cite{kogan2009predicting,qin-yang-2019-say}, defined by:
\begin{equation}
    v_{[t-\tau,t]} = \text{ln}\biggl(\sqrt{\frac{\sum_{i=0}^{\tau}(r_{t-i}-\bar{r})^2}{\tau}}\biggr)
\end{equation}
Here $r_t$ is the return price at day $t$ and $\bar{r}$ the mean of return prices over the period of $t-\tau$ to $t$. We refer to $\tau$ as the temporal volatility window in our experiments. The return price $r_t$ is defined as $r_t = \frac{P_t}{P_{t-1}}-1$ where $P_t$ is the closing price on day $t$.

\begin{table}[h]
    \scriptsize
    \centering
    \begin{tabular}{c|c|c|}
        Group & Man & Woman \\
        \hline
        Count & $333$ & $42$ \\
    \end{tabular}
    \caption{Group distribution in the training set of Earnings Conference Calls}
    \label{tab:dist_ecalls}
\end{table}

\paragraph{Trustpilot}
We only include the datapoints that contains complete demographic attributes, i.e. the gender, age and location, but as with our topic classification experiments, we only study the group that we can define based on age and gender. All attributes are self-reported. For training we divide the reviews into the four resulting groups (\textit{Old-Man}, \textit{Young-Woman}, etc.) and downsample the largest groups to match the size of the smallest group. For validation as well as testing, we withhold 200 samples from each demographic with an even distribution among the ratings (1 to 5).
The review scores are then binarized by grouping positive (4 and 5 stars) and negative (1 and 2 stars) and discarding neutral ones (3 stars). For a similar use of this binarization scheme, see \citet{gupta-etal-2020-effective} and \citet{desai2019}. See the group distributions for the training data in Table \ref{tab:dist_tus} and \ref{tab:dist_tuk} for the US and UK tasks respectively.

\begin{table}[h]
    \centering
    \scriptsize
    \begin{tabular}{c|c|c|c|c|}
        Group & Young, Man & Old, Man & Young, Woman & Old, Woman\\
        \hline
        Count & $7242$ & $7210$ & $7222$ & $7255$ \\
    \end{tabular}
    \caption{Group distribution in the training set of Trustpilot-US}
    \label{tab:dist_tus}
\end{table}
\begin{table}[h]
    \centering
    \scriptsize
    
    \begin{tabular}{c|c|c|c|c|}
        Group & Young, Man & Old, Man & Young, Woman & Old, Woman\\
        \hline
        Count & $18464$ & $18693$ & $18554$ & $18693$ \\
    \end{tabular}
    \caption{Group distribution in the training set of Trustpilot-UK}
    \label{tab:dist_tuk}
\end{table}

\paragraph{BiLSTM}
The BiLSTM model was trained using a Nvidia Tesla K80 GPU. We use a learning rate of $1e^{-2}$ and train using DP-SGD for 30 epochs using a virtual batch size of 32. The average sequence length of the audio embeddings is 159. We set the maximum sequence length to 150 as we did not observe a performance increase for higher values. We run 5 individual seeds for each configuration.

In our differentially private experiments with the BiLSTM (i.e Earnings Conference Calls), we fix the gradient clipping $C$ to $0.8$. By specifying various approximate target levels of $\varepsilon\in\{1,5,10\}$  a corresponding noise multiplier $\sigma$ is computed with the Opacus framework, based on the batch size and number of training epochs. 

\paragraph{DistilBERT}
DistilBERT is a small Transformer model trained by distilling BERT \citep{devlin-etal-2019bert} ({\sf bert-base-uncased}). It has 3/5th of the parameters of {\sf bert-base-uncased}, runs 60\% faster, while preserving over 95\% of the performance of {\sf bert-base-uncased}, as measured on the GLUE language understanding benchmark \citep{wang-etal-2018-glue}.

We finetune DistilBERT on the Trustpilot corpus and Blog Authorship corpus for 20 epochs each, using a batch size of 8, accumulating gradient for a total virtual batch size of 16 using the built in Opcaus functionality. We limit the number of tokens in a sequence to 256 and use a learning rate of $5e^{-4}$ with the AdamW optimizer in addition to a weight decay of $0.01$. Otherwise we use the default parameters defined in the Huggingface Transformers python package (version 4.4.2). 
The models are trained using a single Nvidia TitanRTX GPU and each configuration takes between 5 and 14 hours to run, depending on the size of that dataset and if DP is used or not. We run 3 individual seeds for each configuration.

In our differentially private experiments with DistilBERT (i.e. Blog Authorship and Trustpilot), we fix the gradient clipping $C$ to $1.2$ and by specifying various target levels of $\varepsilon\in\{1,5,10\}$  a corresponding noise multiplier $\sigma$ is computed with the Opacus framework, based on the batch size and number of training epochs. 

\paragraph{Resnet50}
ResNet50 is a variant of the ResNet model \citep{he2015deep}, which has 48 convolution layers along with 1 max pooling and 1 average pooling layer. It has 3.8 x $10^9$ floating points operations.

We finetune our Resnet50 model on the CelebA dataset for 20 epochs using a batch size of 64. We optimize the model using standard stochastic gradient descent (SGD) with a learning rate of $1e^{-3}$, momentum of $0.9$ and no weight decay.
We train our models using a single Nvidia TitanRTX GPU and each configuration takes between 6 and 8 hours to run, depending on if DP is used or not. We run 3 individual seeds for each configuration.

As with the differentially private DistilBERT experiments, we also here fix the gradient clipping $C$ to $1.2$ and by specifying various target levels of $\varepsilon\in\{1,5,10\}$  a corresponding noise multiplier $\sigma$ is computed with the Opacus framework, based on the batch size and number of training epochs. 

\end{document}